\documentclass[12pt,reqno]{amsart}
\usepackage{amsthm,amsfonts,amssymb,euscript,cite}

\theoremstyle{plain}
  \newtheorem{theorem}[subsection]{Theorem}
  
  \newtheorem{proposition}[subsection]{Proposition}
  \newtheorem{lemma}[subsection]{Lemma}
  \newtheorem{corollary}[subsection]{Corollary}
\def\bphi{{\boldsymbol{\phi} }}
\def\bpsi{{\boldsymbol{\psi} }}
\def\F{{\mathcal F}}
\def\R{{\mathbb{R}}}
\def\C{{\mathbb{C}}}
\def\I{{\mathcal I}}

\def\Q{{\mathcal Q}}
\def\ch{\mbox{ch} (k)}
\def\chbb{\mbox{chb} (k)}
\def\trigh{\mbox{trigh} (k)}
\def\p{\mbox{p} (k)}
\def\sh{\mbox{sh} (k)}
\def\shh{\mbox{sh}}
\def\chh{\mbox{ch}}
\def\shhb{\overline {\mbox{sh}}}
\def\chhb{\overline {\mbox{ch}}}
\def\shbb{\mbox{shb} (k)}
\def\chb{\overline{\mbox{ch}  (k)}}
\def\shb{\overline{\mbox{sh}(k)}}
\theoremstyle{remark}
  \newtheorem{remark}[subsection]{Remark}
  
\theoremstyle{definition}

\begin{document}

\def\bphi{{\boldsymbol{\phi} }}
\def\bpsi{{\boldsymbol{\psi} }}
\def\F{{\mathcal F}}
\def\R{{\mathbb{R}}}
\def\C{{\mathbb{C}}}
\def\I{{\mathcal I}}
\def\Q{{\mathcal Q}}
\def\ch{\mbox{\rm ch} (k)}
\def\chbb{\mbox{\rm chb} (k)}
\def\trigh{\mbox{\rm trigh} (k)}
\def\p{\mbox{p} (k)}
\def\sh{\mbox{\rm sh} (k)}
\def\shh{\mbox{\rm sh}}
\def\chh{\mbox{ \rm ch}}
\def\shhb{\overline {\mbox{\rm sh}}}
\def\chhb{\overline {\mbox{\rm ch}}}
\def\shbb{\mbox{\rm shb} (k)}
\def\chb{\overline{\mbox{\rm ch}  (k)}}
\def\shb{\overline{\mbox{\rm sh}(k)}}

\title{Second-order corrections to mean field evolution
 of weakly interacting Bosons. II.}

\author{M. Grillakis}
\address{Department of Mathematics, University of Maryland, College Park, MD 20742}
\email{mng@math.umd.edu}

\author{M. Machedon}
\address{Department of Mathematics, University of Maryland, College Park, MD 20742}
\email{mxm@math.umd.edu}

\author{D. Margetis}
\address{Department of Mathematics\\
and Institute for Physical Science and Technology\\
and Center for Scientific Computation and Mathematical Modeling, University of Maryland, College Park, MD 20742}
\email{dio@math.umd.edu}

\begin{abstract}
We study the evolution of a $N$-body weakly interacting system of Bosons. Our work forms an extension
of our previous paper I \cite{GMM}, in which we derived a second-order correction to a mean-field
evolution law for coherent states in the presence of small interaction potential.
Here, we remove the assumption of smallness of the interaction potential and  prove global existence of
solutions to the equation for the second-order correction. This implies an improved Fock-space estimate for
our approximation of the $N$-body state.
\end{abstract}

\maketitle

\section{Introduction}
\label{sec:intro}

Experimental advances in the Bose-Einstein condensation (BEC) of dilute atomic gases \cite{andersonetal95,davisetal95} have stimulated
interesting questions on the quantum theory of many-body systems.
For broad reviews, see, e.g., \cite{lieb05,pitaevskii03}. In BEC, integer-spin atoms (Bosons) occupy macroscopically
a quantum state (condensate). For a large number $N$ of interacting atoms, the evolution of this system has been described fairly well
by a single-particle nonlinear Schr\"odinger equation \cite{gross61,gross63,pitaevskii61,wuI}. The {\em emergence}
of this {\em mean-field} description from the $N$-body Hamiltonian evolution has been the subject of extensive
studies; see, e.g.,
\cite{E-E-S-Y1,E-Y1,E-S-Y1,E-S-Y2,E-S-Y4,E-S-Y3,Rod-S, K-MMM}.

In \cite{GMM}, henceforth referred to as paper I, we derived a new nonlinear Schr\"odinger equation that describes
a second-order correction to a mean-field approximation for the $N$-body Hamiltonian evolution.
Our work was inspired by: (i) {\em Fock-space} estimates provided by Rodnianski and Schlein \cite{Rod-S}, with regard to
the rate of convergence for Hartree dynamics; and (ii) a second-order correction formulated by Wu \cite{wuI,wuII}, who introduced
a kernel for the scattering of atoms {\em in pairs} from the condensate to other states. In paper I, we derived a new Fock-space estimate; and showed
that for small interaction potential the equation for our second-order correction can be solved locally in time.

The present paper is a continuation of paper I. The main improvement presented here is the removal of our assumption on the smallness
of the interaction potential. Notably, we prove {\em global existence} of solutions to the equation for
the second-order correction. Our approach enables us to derive an improved with respect to time Fock-space estimate for
our approximation of the $N$-body quantum state.

In the remainder of this introduction, we review elements of the Fock space, summarize the major results of paper I, and state
the contributions of the present paper. For a more extensive discussion of the background, the reader may consult, e.g.,
the introduction in our paper I. 

{\bf Fock space and mean field.}
The problem at hand concerns the time evolution of $N$ weakly
interacting Bosons described by
$$
i\partial_{t}\psi =H_{N}\psi~,
$$
where $\psi$ is the $N$-body wave function, $H_N$ the Hamiltonian operator
\begin{align*}
H_{N}:&=\int a^{\ast}_{x}\Delta_{x}a_{x}\ dx -
{1\over 2N}\int v(x-y)a^{\ast}_{x}a^{\ast}_{y}a_{x}a_{y}\ dxdy\\
&=H_0- \frac{1}{N}V~,
\end{align*}
and $v$ is the two-body interaction potential.
A few comments on these expressions are in order. Here, we use the (convenient for our purposes)
formalism of second quantization, where $a^{\ast}$, $a$ are annihilation and creation operators
in a Fock space $\F$ \cite{Berezin}, to be defined below; $\psi$ is a vector in $\F$; and $V$ is the particle interaction.
Note that, in comparison to paper I, we changed the sign of the interaction term $V$, i.e., we
replaced $v$ with $-v$ so that having $v\geq 0$ corresponds to repulsive interaction,
which leads to defocusing behavior.

At this point, it is advisable to review the basics of the Fock space $\F$ over $L^2(\R^3)$.
For Bosons, the elements of $ \F$ are
vectors of the form $\bpsi=(\psi_0, \psi_1(x_1), \psi_2(x_1, x_2), \cdots )$,
where $\psi_0 \in \C$ and $\psi_n \in L^2_s (\R^{3n})$ are symmetric in $x_1, \ldots, x_n$. The Hilbert space structure of $\F$
is given by the inner product $\left(\bphi, \bpsi\right)= \sum_n \int \phi_n \overline{\psi_n} dx$.

For any $f \in L^2(\R^3)$, the (unbounded, closed,  densely defined) creation operator
$a^* (f) :\F \to \F$ and annihilation operator $a(\bar f) : \F \to \F$ are defined by
\begin{equation*}
(a^{\ast}(f)\psi_{n-1})(x_1,\ldots,x_n)=\frac{1}{\sqrt{n}}\sum_{j=1}^nf(x_j)\psi_{n-1}(x_1,\ldots,x_{j-1},x_{j+1},\ldots,x_n)~,
\end{equation*}
\begin{equation*}
(a(\overline f)\psi_{n+1})(x_1,x_2,\ldots,x_n)=\sqrt{n+1}\int\psi_{(n+1)}(x,x_1,\ldots,x_n)\overline f(x)\, dx~.
\end{equation*}
The operator valued distributions $a^*_x$ and $a_x$ are defined by
\begin{align*}
a^*(f) = \int f(x) a^*_x\ dx~,\\
a(\overline f) = \int \overline f(x)\,a_x\ dx~.
\end{align*}
It follows that the operators $a$, $a^{\ast}$ satisfy the commutation relations
$$
[a_{x},a^{\ast}_{y}]=\delta(x-y)~,\qquad
[a_{x},a_{y}]=[a^{\ast}_{x},a^{\ast}_{y}]=0\ .
$$

We are interested in the evolution of coherent states, i.e., vectors of the form
$e^{-\sqrt{N}A(\phi)}\Omega$ where $\Omega=(1,0,\ldots)\in \F$ is the vacuum state, $\phi(t,x)$ is the one-particle
wave function (to be determined later), and
\begin{equation}
A(\phi):=\int \big(\phi(x)a^{\ast}_{x}-\overline{\phi}(x)a_{x}\big)\ dx\ .
\label{eq:A-def}
\end{equation}
It is important to notice that
$$
e^{-\sqrt{N}A(\phi)}\Omega =
\big(\ldots c_{n}\prod_{j=1}^{n}\phi(x_{j})\ldots \big)~.
$$
Thus, the $n^{th}$ slot in the coherent state Fock vector consists of the tensor product of
$n$ functions $\phi(x)$; the relevant constant is $c_n=\left(\frac{N^n}{ n!} \right)^{1/2}$.

Furthermore, the number operator, ${\mathcal N}:=\int a^{\ast}_{x}a_{x}\ dx$, satisfies
$$
\big<\Omega e^{\sqrt{N}A(\phi)}\ \big\vert\ {\mathcal N}\ \big\vert\
e^{-\sqrt{N}A(\phi)}\Omega \big> =N\Vert\phi\Vert^{2}~.
$$
Thus, if we normalize the wave function by setting $\Vert\phi\Vert =1$,
the average number of
particles remains constant, $N$.

It can be claimed that a reasonable approximation for the many-body time evolution is
expressed by the Fock vector
$$
\psi_{appr}:=e^{\sqrt{N}A(\phi(t))}\Omega~,
$$
where $\phi(t,x)$ satisfies the Hartree equation \eqref{H-F-phi}.
This $\psi_{appr}$ encapsulates the mean field approximation for $N$ weakly interacting Bosons.
The precise meaning of this approximation as well as its rigorous justification
were studied within the PDE setting by Erd\"os, Schlein, Yau \cite{E-Y1,E-S-Y1,E-S-Y2,E-S-Y4,E-S-Y3} via Bogoliubov-Born-Green-Kirkwood-Yvon hierarchies
for reduced density matrices (see also  Klainerman and Machedon,  \cite{K-MMM}, for a simplification of the uniqueness part of the argument). In the Fock space setting, the mean field approximation was studied by Ginibre and Velo \cite{G-V}
and, most recently, by Rodnianski and Schlein \cite{Rod-S}; see also Hepp, \cite{hepp}.

{\bf Main results of paper I.}
Starting with a coherent state as initial data, in \cite{GMM} we proposed a correction of
the form
$$
\psi_{appr}:=e^{-\sqrt{N}A(\phi)}e^{-B(k)}\Omega~,
$$
where
\begin{equation}
B(k):=\int \big(k(t,x,y)a^{\ast}_{x}a^{\ast}_{y}-\overline{k}(t,x,x)a_{x}a_{y}\big)
\ dxdy
\label{eq:B-def}
\end{equation}
and the kernel $k(t,x,y)$ satisfies an appropriate nonlinear evolution equation.
This $k$ loosely corresponds to the ``pair excitation function'' introduced by Wu \cite{wuI,wuII} but our
set-up and derived equation for $k$ are different.
By assuming that the interaction potential $v(x-y)$ is small,
we proved that for a {\em finite time interval} our approximation stays close to the
exact evolution in the Fock space norm. To be more precise, we proved the following general theorem.

\begin{theorem} \label{main}
Suppose that $v$ is an even
potential.
Let $\phi$ be a smooth solution of the Hartree equation
\begin{align}
i \frac{\partial \phi}{\partial t} + \Delta \phi \label{H-F-phi}
- (v*|\phi|^2)\phi =0
\end{align}
with initial conditions $\phi_0$. Assume all conditions (1)-(3) listed below:
\begin{enumerate}
\item \label{c1}
The kernel $k(t, x, y) \in L^2(dx dy)$ for all $t$, is symmetric,
and solves the equation
\begin{align} \label{newnlsshort}
(i u_t + u g^T + g u - (1+p)m)  = (i p_t + [g, p] + u \overline m)(1+p)^{-1} u~,
\end{align}
where all products in \eqref{newnlsshort} are interpreted as spatial compositions of kernels (or operator products), ``$1$" is the identity operator, and
\begin{subequations} \label{defs}
\begin{align}
&u(t, x, y):= \sh:=k + \frac{1}{3!} k \overline k k + \ldots~, \\
&\delta(x-y)+p(t, x, y):= \ch:=\delta(x-y) + \frac{1}{2!} k \overline k + \ldots~,\\
&g(t, x, y):= - \Delta_x \delta (x-y)
+v(x-y) \phi(t, x) \overline \phi(t, y)  \\
& \hphantom{g(t, x, y):=}  +(v * |\phi|^2 )(t, x) \delta(x-y)~, \\
&m(t, x, y) := -v(x-y)\overline \phi(t, x) \overline \phi(t, y)~.
\end{align}
\end{subequations}
\item \label{c2}
The functions defined by
\begin{align*}
f(t):= \|e^B[A, V]e^{-B} \Omega\|_{\F}~,
\end{align*}
\begin{align*}
g(t):=\|e^BV e^{-B}\Omega\|_{\F}~,
\end{align*}
are locally integrable; recall that $V$ is the interaction operator, and $A$, $B$ are operators defined by \eqref{eq:A-def}, \eqref{eq:B-def}.
\item \label{c3}
The trace $\int d(t, x, x)\ dx$ is locally integrable in time, where the kernel $d(t,x,y)$ is
 \begin{align*}
d(t, x, y)=
&
\left(i \sh_t +\sh g^T+ g \sh\right) \shb \notag\\
-&\left(i \ch_t + [g, \ch]\right)\ch\notag\\
-&\sh \overline m \ch -\ch m \shb~.
\end{align*}
\end{enumerate}
Then, there exist real functions $\chi_0$,  $\chi_1$ such that
\begin{align}
&\| e^{-\sqrt N A(\phi(t))} e^{-B(\phi(t))} e^{ -i \int_0^t ( N \chi_0(s) + \notag \chi_1(s))ds} \Omega - e^{itH_N} \bpsi_0 \|_{\F} \\
&\le
\frac{\int_0^t f(s) ds}{\sqrt N} +
\frac{\int_0^t g(s) ds}{ N}~.\label{l2est}
 \end{align}
\end{theorem}

Moreover, we showed that the hypotheses of this theorem are satisfied locally in time if $v$ is small.

\begin{theorem} \label{localex-I}
Let $\epsilon_0$ be sufficiently small and
$v(x)= \chi(x) \frac{\epsilon_0}{|x|}$ for $\chi \in C_0^{\infty}(\Bbb R^3)$ . Assume that $\phi$ is a smooth solution to
the Hartree equation \eqref{H-F-phi}, $\|\phi\|_{L^2(dx)}=1$. Then, there exists
$k \in L^{\infty}([0, 1]) L^2 (dx dy)$ that solves \eqref{newnlsshort} for $0 \le t \le 1$ with initial condition $k(0, x, y)=0$.
In addition, we have the estimates
\begin{align*}
\int _0^1\|e^B V e^{-B} \Omega\|^2_{\F} \, \, dt\le C~,
\end{align*}
and
\begin{align*}
\int _0^1\|e^B [A, V] e^{-B} \Omega\|^2_{\F} \, \, dt\le C~.
\end{align*}

\end{theorem}

{\bf Main results of this paper.}
In the present paper, we remove the smallness assumption on the
interaction potential, prove that the evolution equation of $k(t,x,y)$ has a
global in time solution {\it and obtain a stronger a priori estimate for the
difference of the approximate and exact solution for the $N$-body Fock-space vector}.
In particular, we prove the following theorem.

\begin{theorem} \label{mainnewthm}
Let the notation be as in Theorem \ref{main}.
Consider $v(x) =\frac{\chi(x)}{|x|} \ge 0$, where $\chi \in C_0^{\infty}$ and $\chi(r)$ is a decreasing cut-off function.
Assume $\phi_0$ has sufficiently many
derivatives in $L^2$ and $\|x\phi_0 \|_{L^2} \le C$. Further, suppose $k(0, \cdot, \cdot)\in L^2(\Bbb R^6)$ is prescribed.
Then, the hypotheses of Theorem \ref{main} are satisfied globally in time and
\begin{align}
&\| e^{-\sqrt N A(\phi(t))} e^{-B(\phi(t))} e^{ -i \int_0^t ( N \chi_0(s) + \notag \chi_1(s))ds} \Omega - e^{itH_N} \bpsi_0 \|_{\F} \\
&\le C
\frac{(1+t)^\frac{1}{2}}{\sqrt N}
~.\label{newl2est}
 \end{align}
\end{theorem}

\begin{remark}
It follows  from our calculations that
if we omit the assumption $v \ge 0$, the hypotheses of
Theorem \ref{main} are still satisfied globally in time, but we no longer have estimate \eqref{newl2est}.
\end{remark}

The remainder of this paper is largely devoted to the proof of Theorem \ref{mainnewthm} and is organized as follows.
In section \ref{apr} we derive the a priori estimate
\begin{align*}
\Vert u(T)\Vert_{L^{2}} \leq  \left(\int_{0}^{T}  \Vert  m\Vert_{L^{2}}dt + \|u(0)\|_{L^2} \right)
\exp\left(\int\limits_{0}^{T} \Vert m\Vert_{L^{2}}dt\right)\ .
\end{align*}
In section \ref{pseudo} we prove that $\int_{0}^{\infty}  \Vert  m\Vert_{L^{2}}dt \le C$ if $v \ge 0$.
In section \eqref{locex} we show that
\eqref{newnlsshort} is locally well posed for $L^2$, possibly large, initial conditions for $u$. This proof is much harder than the corresponding one in
paper I; the latter worked for zero (or small) $L^2$ initial conditions.
The idea here is to transform the quasilinear equation \eqref{newnlsshort} into an equivalent semilinear one.
Section \ref{er} is devoted to estimating the error terms $f$ and $g$ entering \eqref{l2est}.
In section \ref{expb} we construct the requisite operator $e^{B}$ in the case where
$\|k\|_{L^2}$ is large and $e^B$ is no longer defined as a convergent Taylor series; and elaborate on the connection
of this construction with the Segal-Shale-Weil, or metaplectic, representation.
Finally, the Appendix focuses on an improved computation of some error terms previously computed in section 8 of paper I.
This leads to a simpler proof of our stronger estimate \eqref{newl2est}.
Our notation is not uniform across sections, but is self-explanatory and convenient. When the variables are called $x_1$ and $x_2$, $\phi_1$ abbreviates $\phi(x_1)$, $v_{1-2}=v(x_1 - x_2)$, etc.

\section{Pseudoconformal transformation for Hartree equation \label{pseudo}}

The goal of this section is to find an estimate for the decay rate in time of $\|\phi(t, \cdot) \|_{L^4(\Bbb R^3)}$,
where $\phi$ is a solution of the Hartree equation,
\begin{align}
i \frac{\partial \phi}{\partial t} + \Delta \phi \label{heq}
- (v*|\phi|^2)\phi =0~,
\end{align}
with initial condition $\phi_0$ such that
$\|\phi_0\|_{H^1} +
\|x \phi_0\|_{L^2}
$ be finite. For this purpose, we make use of the technology of dispersive estimates from \cite{G-M}

We start with some preliminaries. Let
\begin{equation}
W=v * |\phi|^2~.
\label{W-def}
\end{equation}
The quantities relevant for the conservation laws (to be stated below) are defined by
\begin{align*}
&\rho :=(1/2)|\phi|^2~; \\
&p_{j}:=(1/2i)\left(\phi\nabla_{j}\overline \phi -\overline \phi\nabla_{j}\phi\right)~;
\quad p_{0}=(1/2i)\left(\phi\partial_{t}\overline{\phi } -\overline{\phi}\partial_{t}\phi\right)~; \\
&\sigma_{jk}:=\nabla_{j}\overline{\phi} \nabla_{k}\phi +\nabla_{k}\overline{\phi}\nabla_{j}\phi~;
\quad \sigma_{0j}=\nabla_{j}\overline{\phi}\partial_{t}\phi +\partial_{t}\overline{\phi}\nabla_{j}\phi~.
\end{align*}
Let us call
$\sigma :=tr(\sigma_{jk})$, the trace of the tensor $\sigma_{jk}$, and define two more quantities, namely,
\begin{align*}
\lambda :&=-p_{0}+(1/2)\sigma +W\rho=\Delta \rho - W \rho~;\\
e:&=(1/2)\sigma +W\rho~.
\end{align*}
With regard to $\lambda$, see \eqref{infinitesimal}.

The quantity $e$ is the energy density, while $\lambda$ is the Lagrangian density. Indeed, one can see that the evolution
equation can be derived as a variation of the integral
\begin{align*}
{\mathcal L}(\phi ,\overline{\phi}):=\int dx\left\{\lambda\right\}\ .
\end{align*}
The associated conservation laws can be stated in the forms
\begin{subequations}
\label{conserv-laws}
\begin{align}
&\partial_{t}\rho -\nabla_{j}p^{j}=0~,\label{conserv-a} \\
&\partial_{t}p_{j} -\nabla_{k}\left\{\sigma_{j}^{\ k}-\delta_{j}^{\ k}\lambda\right\} +l_{j}=0~, \label{conserv-b}\\
&\partial_{t}e -\nabla_{j}\sigma_{0}^{\ j}+l_{0} =0~.\label{conserv-c}
\end{align}
\end{subequations}
These laws express the conservation of mass, momentum and energy, respectively,\footnote{
As it is well known, any Euler-Lagrange equation derived from a local Lagrangian has a conserved energy-momentum tensor; see, for instance,
section 37.2 in \cite{DFN}. In our case, $T_{jk}=\sigma_{jk} - \delta_{jk} \lambda$,
$T_{j0}=-p_j$, $T_{0j}=\sigma_{0j}$ and $T_{00}=-e$. The vectors $l_0$, $l_j$ are corrections due to the fact that our Lagrangian is nonlocal.}
where the vector $\big(l_{j}, l_{0}\big)$ is
\begin{align*}
l_{j}:=W\rho_{,j}-W_{,j}\rho\ ,\qquad l_{0}:=W\rho_{,t}-W_{,t}\rho\ .
\end{align*}
We can see that the momentum and energy are indeed conserved quantities: $l_{0}$ and $l_{j}$ average to zero, since
$\int (v * \partial \rho) \rho =\int (v * \rho) \partial \rho$
for an even $v$.

One can derive one more identity (a structure equation) by
multiplying the evolution equation by $\overline{\phi}$
and taking the real part:
\begin{align}
\lambda +\big(-\Delta\rho +W\rho\big) =0\ . \label{infinitesimal}
\end{align}
This equation is the result of an infinitesimal transformation on the target when we regard $\phi$ as a map into the complex
plane. Using the structure equation, we can recast the conservation of momentum, equation \eqref{conserv-b}, into the form
\begin{align*}
\partial_{t}p_{j} -\nabla_{k}\left\{\sigma_{j}^{\ k}+\delta_{j}^{\ k}\big(-\Delta\rho +W\rho\big)\right\} +l_{j}=0\ .
\end{align*}

Let us return to conservation laws \eqref{conserv-laws}.
The conformal identity can be derived by contracting the mass equation \eqref{conserv-a} with $\vert x\vert^{2}/2$; the momentum equation \eqref{conserv-b}
with $tx^{j}$; and the energy equation \eqref{conserv-c} with $t^{2}$; and adding the resulting identities.
The final result can be written in the abstract form
\begin{align}
\partial_{t}e_{c} -\nabla_{j}\tau^{j} +r =0\ ,
\label{ec-evolution}
\end{align}
where the relevant quantities are
\begin{align*}
&e_{c}:=(\vert x\vert^{2}/2)\rho +tx^{j}p_{j}+t^{2}e =t^{2}\left(\big\vert\nabla \left(e^{i\vert x\vert^{2}/4t}\phi\right)\big\vert^{2} +W\rho\right)~,
\\
&\tau^{j}:=(\vert x\vert^{2}/2)p^{j}+tx^{k}\sigma^{j}_{\ k}+tx^{j}\big(-\Delta\rho +W\rho\big) +t^{2}\sigma_{0}^{\ j}~, \\
&r:=t^{2}l_{0} +tx^{j}l_{j} -nt\Delta\rho +t(n-2)W\rho~,
\end{align*}
By integrating \eqref{ec-evolution} in space, we obtain the ODE
\begin{align*}
\dot{E}_{c}+R_{c} =0~,
\end{align*}
where
\begin{eqnarray}
E_{c}&:=&\int dx\left\{e_{c}\right\}~, \label{ec} \\
R_{c}&:=&\int dx\left\{(n-2)tW\rho +tx^{j}l_{j}\right\}\ \label{rc-def};
\end{eqnarray}
note that $E_{c}$ is the pseudoconformal energy. Next, we recast $R_c$ into a convenient form.
By inspection of \eqref{rc-def}, it remains to compute
\begin{align*}
&\int  x^{j}l_{j} dx\\
&=\int dx\left\{ Wx^{j}\rho_{,j}-x^{j}W_{,j}\rho\right\} =2\int dx_{1}dx_{2}\left\{
v_{1-2}\big(x_{1}\cdot \nabla_{1-2}(\rho_{1}\rho_{2})\right\} \\
&=\int dx_{1}dx_{2}\left\{v_{1-2}\Big((x_{1}+x_{2})\cdot\nabla_{1-2}
+(x_{1}-x_{2})\cdot\nabla_{1-2}\Big)(\rho_{1}\rho_{2})\right\} \\
&=\int dx_{1}dx_{2}\left\{\Big(-2nv_{1-2}-2(x_{1}-x_{2})\cdot\nabla v_{1-2}\big)(\rho_{1}\rho_{2})\right\}\ .
\end{align*}

In the above calculation, we used the fact that $v_{1-2}$ is symmetric with respect to the $1\to 2$ and $2\to 1$ transposition,
while $\nabla_{1-2}$ is antisymmetric. Moreover, we have
\begin{align*}
\nabla_{1-2}\cdot (x_{1}-x_{2})=2n\ ;\quad
(x_{1}-x_{2})\cdot \nabla_{1-2}v_{1-2}=2(x_{1}-x_{2})\cdot \nabla v_{1-2}\ .
\end{align*}
Finally, regarding \eqref{rc-def}, notice that
\begin{align*}
\int dx\left\{(n-2)W\rho\right\} =
\int dx_{1}dx_{2}\left\{2(n-2)v_{1-2}\rho_{1}\rho_{2}\right\}\ .
\end{align*}
Substituting back into \eqref{rc-def}, we wind up with the integral
\begin{align}
R_{c}=t\int dx_{1}dx_{2}\left\{\big[(-4)v_{1-2}-2r_{1-2}v^{\prime}_{1-2}\big](\rho_{1}\rho_{2})\right\}\ . \label{rc}
\end{align}
This integral is used as an alternate expression for $R_c$.

Thus, we have proved the following lemma.

\begin{lemma} \label{pseudo1} Let $\phi$ be a solution of the Hartree equation \eqref{heq},
and let $E_c$, $R_c$ be defined by
\eqref{ec}, \eqref{rc}. Then, the following equation holds:
\begin{align}
\dot{E}_{c}+R_{c} =0 \ .\label{conf}
\end{align}
\end{lemma}

\begin{remark} In order to obtain a {\em decreasing} pseudoconformal energy, we need $R_c \ge 0$, which is unfortunately {\em not true} for the Coulomb potential.
Instead, we proceed to show that $R_c$ is integrable in time.
 \end{remark}
We first state another consequence of our previous calculations.

\begin{lemma}
Define
\begin{align*}
E_{cc}:= \frac{E_c}{t}=
\int dx\left\{t\left(\big\vert\nabla \left( e^{i\vert x\vert^{2}/4t}\phi\right)\big\vert^{2}+W\rho\right)\right\}\ .\\
\end{align*}
Then, $E_{cc}$ satisfies
\begin{align*}
\dot{E}_{cc}+R_{cc} =0\ ,
\end{align*}
where $R_{cc}$ is defined by
\begin{align*}
R_{cc}:=\big\vert\nabla \left(e^{i\vert x\vert^{2}/4t}\phi \right)\big\vert^{2}
-2 \int \left( v + r_{1-2} v'_{1-2} \right) \rho_1\rho_2\ .
\end{align*}
\end{lemma}
\begin{remark} Notice that for $v(x) = \chi(|x|) \frac{1}{|x|}$, $R_{cc}$ is {\em positive} if $\chi(r)$ is
decreasing for $r>0$; thus, $E_{cc}$ is decreasing. \end{remark}
In conclusion, using the Sobolev embedding and interpolation we have the following corollary.
\begin{corollary}
Let $\phi$ be a solution of the Hartree equation \eqref{heq}. Then, the following estimates
hold for all $t\ge 1$:
\begin{align}
&\|\phi(t, \cdot)\|_{L^6(\Bbb R^3)} \le \frac{C}{\sqrt t} E_{cc}(1)\ , \notag\\
&\|\phi(t, \cdot)\|_{L^4(\Bbb R^3)} \le \frac{C}{ t^{3/8}} E_{cc}(1)\ . \label{l4e}
\end{align}
\label{cor-Ecc}
\end{corollary}

Using Lemma \eqref{pseudo1}, the result of Corollary \ref{cor-Ecc} can be improved:
\begin{theorem} \label{hest}
Let $\phi$ be a global smooth solution of the Hartree equation
\begin{align}
i \frac{\partial \phi}{\partial t} + \Delta \phi
- (v*|\phi|^2)\phi =0 \label{har}
\end{align}
with initial condition $\phi_0$ such that $E_c(1)$ is finite.
Then,
\begin{align*}
&\| \phi(t, \cdot)\|_{L^6}
\le  C t^{-3/4}\ ,\\
&\| \phi(t, \cdot)\|_{L^4} \le C t^{-9/16}\ ,
\end{align*}
and, thus,
\begin{align*}
&\int_1^{\infty} \|\phi(t, \cdot) \|^2_{L^6(\Bbb R^3)} dt \le C\ ,\\
&\int_1^{\infty} \|\phi(t, \cdot) \|^2_{L^4(\Bbb R^3)} dt \le C\ .
\end{align*}
\end{theorem}
\begin{proof}
Using the fact that
$-4v-2rv^{\prime}\in L^{1}$ together with \eqref{l4e}, we see that
\begin{align*}
 R_c(t) \le C t \|\phi(t, \cdot)\|^4_{L^4} \le C t^{-1/2}\ .
 \end{align*}
By integrating \eqref{conf}, we conclude that
\begin{align*}
E_c(t) \le C t^{1/2}
\end{align*}
for $t \ge 1$.
Using the Sobolev inequality and the definition of $E_c$ (see \eqref{ec}) we conclude that
\begin{align*}
\| \phi(t, \cdot)\|_{L^6} \le C \| \nabla
\left( e^{i\vert x\vert^{2}/4t}\phi \right)\|_{L^2}
\le  C t^{-3/4}\ .
\end{align*}
Interpolation with energy conservation gives
\begin{align*}
\| \phi(t, \cdot)\|_{L^4} \le C t^{-9/16}\ .
\end{align*}
\end{proof}

\section{ A priori Estimates \label{apr} }

In this section, by using Theorem \ref{hest} we derive a priori estimates for the solution
$u$ of \eqref{newnlsshort}.
We recall definitions \eqref{defs} of Theorem \ref{main}, suitably  abbreviated,
and add a few new ones:
\begin{align*}
&m_{12}
:=-v_{1-2}\overline{\phi}_{1}\overline{\phi}_{2} =\left (-v(x-y)
\overline{\phi}(x) \overline{\phi}(y)\right)~,
\\
&g_{12}:=-\Delta_{1}\delta_{12}+w_{12}~,\\
& w_{12}:=-v_{1-2}\phi_{1}\overline{\phi}_{2}~,\\
&\sh := u~,\\
&\ch := 1 +p := \delta_{12} +p~,\\
&r := (1+p)^{-1}~,\\
&q:= u \overline u~.
\end{align*}

These are all operators kernels, and their products are interpreted as compositions.
Notice that $w$ and $m$ have the symmetries
$w_{21}=\overline{w}_{12}$, i.e., $w^{*}=w$; and $m_{21}=m_{12}$, i.e., $m^{T}=m$.
The evolution equation for $u = \sh$, given by \eqref{newnlsshort}, is  abbreviated to
\begin{align}
S(u) - (1+p)m= \left(W(p) + u \overline m \right) r u\ , \label{seq}
\end{align}
where
\begin{align*}
&S(u):=iu_{t}+gu +u g^T\ ,\\
&W(p):= i p_t + [g, p]\ ,
\end{align*}
and $u_{12}$ is symmetric,
$u_{21}=u_{12}$, i.e., $u^{T}=u$, while $p_{12}$ is self-adjoint, $p_{21}=\overline{p}_{12}$, i.e., $p^{*}=p$.

Notice that $q$ is related to $p$ by
\begin{align}
q = 2p + p^2\ . \label{trig}
\end{align}
Trigonometric identities such as \eqref{trig}
follow from $e^K e^{-K} = I$
for
\begin{align*}
K=\left(
\begin{matrix}
0 & k\\
\overline k & 0
\end{matrix}
\right)\ .
\end{align*}
The key observation in this section is the following lemma.

\begin{lemma} The following identities hold:
\begin{align}
&\left(W(p)+u\overline{m}\right)r+r\left(W(p)-m\overline{u}\right)=0 \ ,\label{lem}\\
&F:=W(q) = m \overline{u}(1+p) - (1+p) u \overline{m}\ . \label{lem1}
\end{align}
These equations are equivalent for any positive semi-definite kernel $p$,
$q = p^2 + 2p$, and $r=(1+p)^{-1}$.

\end{lemma}
\begin{proof}
To prove \eqref{lem}, multiply \eqref{seq} on the right  by $\overline u$, take the adjoint of \eqref{seq}, namely,
\begin{align*}
\overline{S}(\overline{u})-\overline{m}(1+p)=\overline{u}r
\left(-W(p)+m\overline{u}\right)\\ ,
\end{align*}
multiply it on the left by $u$, and then subtract.
The resulting equation reads
\begin{align*}
W(q)&= W(p)rq+qrW(p) +u\overline{m}rq-qrm\overline{u} \\
&+(1+p)m\overline{u}-u\overline{m}(1+p)\ .
\end{align*}
Now we can combine two terms as follows, using the fact that $r=(1+p)^{-1}$ and $q=(1+p)^2 -1$:
\begin{align*}
u\overline{m}rq-u\overline{m}(1+p)=
u\overline{m}\left((1+p)^{-1}\big[(1+p)^{2}-1\big]-(1+p)\right) =-u\overline{m}r\ .
\end{align*}
Using the equation $q = 2p + p^2$, we obtain
\begin{equation*}
W(2p+p^{2})+u\overline{m}r -rm\overline{u} =W(p)r(2p+p^{2}) +(2p+p^{2})rW(p)\ ;
\end{equation*}
hence,
\begin{align*}
&2W(p)+pW(p)+W(p)p +u\overline{m}r -rm\overline{u}\cr
&=W(p)(1+p)^{-1}\big[1+(1+p)\big]p +
p\big[1+(1+p)\big](1+p)^{-1}W(p)\ .
\end{align*}
Thus, we have
\begin{equation*}
2W(p) +u\overline{m}r -rm\overline{u}= W(p)rp+prW(p) \ .
\end{equation*}
Now observe that $1-rp=r$ to conclude the proof of \eqref{lem}.
To prove \eqref{lem1}, multiply \eqref{lem} on the right and left by $(1+p)$
and recall  $q = 2p + p^2$.
\end{proof}

We are ready to state and prove our main a priori estimate:
\begin{theorem} \label{thm:apriori}
Let $u = \sh$ be a solution of \eqref{newnlsshort} on some interval
$[0, T]$. Then,
the following estimate holds:
\begin{align}
\Vert u(T)\Vert_{L^{2}} \leq  \left(\int_{0}^{T}  \Vert \label{apriori} m\Vert_{L^{2}}dt + \|u(0)\|_{L^2} \right)
\exp\left(\int\limits_{0}^{T} \Vert m\Vert_{L^{2}}dt\right)\ .
\end{align}
\end{theorem}

\begin{proof}
Taking the trace in \eqref{lem1} we obtain
\begin{align}
{d\over dt}\Vert u\Vert^{2}_{L^{2}} = \label{100}
tr\left[(1/i)\big( m \overline{u}(1+p) - (1+p) u \overline m \big)\right] \ .
\end{align}
Thus, we have
\begin{align*}
&{d\over dt}\Vert u\Vert^{2}_{L^{2}}
\le 2 \left( \|m\|_{L^{2}} \|u\|_{L^{2}} + \|m\|_{L^{2}} \|u\|_{L^{2}} \|p\|_{L^{2}} \right)\\
&\le 2 \left( \|m\|_{L^{2}} \|u\|_{L^{2}} + \|m\|_{L^{2}} \|u\|^2_{L^{2}} \right)~.
\end{align*}
The inequality $\|p\| \le \|u\|$ follows by talking the trace
of \eqref{trig} together with the observation that
$tr (p) \ge 0$.
Now we can employ a Gronwall type inequality to deduce \eqref{apriori}.

\end{proof}

Summarizing the results of the previous two sections, we draw our main conclusion.

\begin{corollary}

Let $\phi$ be a solution of the Hartree equation satisfying the assumptions of  Theorem \ref{hest} and let $u = \sh$ be a solution of equation \eqref{newnlsshort} on $[0, T]$, as in Theorem
\ref{thm:apriori}. Assume the potential $v$ is in $L^2(\Bbb R^3)$. Then the following estimate holds:
\begin{align}
\Vert u(T)\Vert_{L^{2}} \leq C\left(1+ \Vert u(0)\Vert_{L^{2}}\right)~.
\end{align}

\end{corollary}

\begin{proof}
By Theorem \ref{thm:apriori}, it suffices to control
$\|m\|_{L^1(dt) L^2 (dx \, dy)}$. Notice that $\|m\|^2_{L^2(\Bbb R^6)} = \int \left(v^2 * |\phi|^2 \right) |\phi|^2 dx \le C \|v^2\|_{L^1(\Bbb R^3)}
  \|\phi\|^4_{L^4(\Bbb R^3)}$. Using the estimates of
 Theorem \ref{hest}, we conclude that $\|m\|_{L^1(dt) L^2 (dx \, dy)} \le C$.

\end{proof}

\section{The local existence Theorem for equation \eqref{newnlsshort} \label{locex}}

In paper I, we showed that \eqref{newnlsshort} has local solutions provided $u(0)=0$ and $v(x) =
\epsilon \frac{\chi(x)}{|x|}$ for  $ \chi \in C_0^{\infty}$.
In this section we relax the assumptions to $u(0) \in L^2(\Bbb R^6)$
and  $v(x) =
 \frac{\chi(x)}{|x|}$ and prove local existence in an interval where
 $\|\phi\|_{L^2([0, T]) L^4 (dx \, dy)}$ is small. Notice that by Theorem
 \ref{hest},
 $[0, \infty)$ can be divided into finitely many such intervals.
 This implies global existence for equation \eqref{newnlsshort}.

 In this setting, we can no longer assume that $\|u\|_{L^{\infty} L^2}$ is small, and terms such as $W(p) r u$ are no longer small compared to $S(k)$
 (see  \eqref{seq} for the notation). Our equation  seems quasilinear, but
 can be transformed into a semilinear one.
 In order to prove local existence, we must solve for $u = \sh$ rather than $k$, and express $p= \sqrt{1+ u \overline u}$ in the operator sense.
Thus we have to prove the following proposition:

\begin{proposition}
The map
\begin{align*}
k \mapsto \sh=u
\end{align*}
is one to one, onto, continuous, with a continuous inverse, from symmetric  Hilbert-Schmidt kernels $k$ onto
symmetric  Hilbert-Schmidt kernels $u$.
\end{proposition}

\begin{proof}
The appropriate context for this proof is set by noticing that the equation $u = \sh$ is equivalent to
\begin{align*}
\exp
\left(
\begin{matrix}
0 & k\\
\overline k &0
\end{matrix}
\right)
=
\left(
\begin{matrix}
\sqrt{1 + u \overline u} & u\\
\overline u & \sqrt{1 + \overline u u}
\end{matrix}
\right)
\end{align*}
By the spectral theorem,
the exponential map is a continuous bijection from
 self-adjoint Hilbert-Schmidt
 "matrices" to positive definite
"matrices" $P$ for which $\|I-P\|_{L^2}$ is finite.
Our target matrix is
\begin{align*}
P=\left(
\begin{matrix}
\sqrt{1 + u \overline u} & u\\
\overline u & \sqrt{1 + \overline u u}
\end{matrix}
\right)~.
\end{align*}
Besides being positive definite, this matrix is symplectic; thus, it satisfies
$P^T J P=J$ where
\begin{align*}
J=\left(
\begin{matrix}
0 & -1\\
1 & 0
\end{matrix}
\right)~,
\end{align*}
and also satisfies $L P L = P^{-1}$ where
\begin{align*}
L=\left(
\begin{matrix}
1 & 0\\
0 & -1
\end{matrix}
\right)~.
\end{align*}

Thus, we have $P = e^p$ where $p$ is self-adjoint.
Since $e^{p^T}Je^p=J$, we conclude that $p$ is
 symplectic, or
$p^TJ+Jp=0$. (Proof: $e^{JpJ}J=Je^{-p}J$ is always true; thus, by easy algebra $e^{p^T}=e^{JpJ}$. Since both $p^T$
and $JpJ$ are self-adjoint, the exponential is one-to-one, and we conclude
that $p^TJ + Jp=0$.)
Similarly, from $L e^p L=e^{-p}$ we infer
$LpL =-p$. The first two conditions force $p$ to be
 of the form
\begin{align*}
p=\left(
\begin{matrix}
a & b\\
c & -a^T
\end{matrix}
\right)~,
\end{align*}
where $a=a^*$, $b=b^T$, $c=b^*$. The third condition entails $a=0$.
Thus, $p$ can be re-written as
\begin{align*}
p=\left(
\begin{matrix}
0 & k\\
\overline k & 0
\end{matrix}
\right)~.
\end{align*}

\end{proof}

The main {\em new} ingredient of this section is the following theorem.
\begin{theorem} \label{equiv}
The following equations are equivalent for a symmetric, Hilbert-Schmidt $u$:
\begin{align}
S(u) &= (1+p) m + \left(W(p) + u \overline m \right) r u\ ; \label{seq1}\\
S(u)&=(1 + p) m +
\frac{1}{2}[W(p), r] u
+\frac{1}{2}\left(r m \overline u + u \overline m r \right) u\ ;
 \label{seq2}\\
S(u)&=(1+p)m +\frac{1}{2}[W, r]u +\frac{1}{2}\left(r m \overline u + u \overline m r \right) u\ ;\label{seq3}
\end{align}
where we set
\begin{align*}
&F:=m\overline{u}(1+p)-(1+p)u\overline{m}~, \\
&W:={1\over 2\pi i}
\int\limits_{\Gamma}\big(q-z\big)^{-1}F\big(q-z\big)^{-1}\sqrt{1+z}\,dz~,\\
&q:=u \overline{u}~,\\
&1+ p:=\sqrt{1+u \overline{u}}~,\\
&r:= (1 + u \overline u)^{-1}~.
\end{align*}
Here, $\Gamma$ is a contour enclosing the spectrum of the non-negative Hilbert-Schmidt operator $u \overline{u}$. Equation \eqref{seq1} is the same as
\eqref{newnlsshort}, suitably re-written. Note that $F$ corresponds to $W(q)$ and $W$ corresponds to $W(p)$.
\end{theorem}

\begin{proof}
Assume $u$ satisfies \eqref{seq1}. Recalling the estimate
\eqref{lem} we conclude $u$ satisfies
\begin{align*}
S(u) &=\frac{1}{2}\left( W(p) r + r W(p) \right) u
+\frac{1}{2}\left( W(p) r - r W(p) \right) u\\
&+(1+p)m  + u \overline{ m} r u\\
&=\frac{1}{2}\left(r m \overline u - u \overline m r \right) u
+\frac{1}{2}[W(p), r] u
+(1+p)m  + u \overline{ m} r u\\
&=\frac{1}{2}[W(p), r] u
+\frac{1}{2}\left(r m \overline u + u \overline m r \right) u
+(1 + p) m~.
\end{align*}
Thus, $u$ satisfies \eqref{seq2}. Notice that both \eqref{seq2} and \eqref{seq3}
are of the form
\begin{align}
S(u)=Xu + (1+p)m~, \label{seq4}
\end{align}
where $X$ is self-adjoint. To see that $X$ is self-adjoint, notice that both $W(p)$ and $W$ are skew-Hermitian. Then,
the procedure can be reversed to show that if $u$ satisfies \eqref{seq4}
then the identity \eqref{lem1}, and thus \eqref{lem}, holds.

Indeed, composing the complex conjugate of \eqref{seq2} on the left with $u$,
we obtain
\begin{align*}
u \overline{S(u)} =u \overline{Xu} + u \overline{(1+p)m}~.
\end{align*}
The adjoint of this operator is
\begin{align*}
S (u) \overline{u} =u \overline{Xu} + m  \overline{(1+p) u}
\end{align*}
Subtracting the first equation from the second one gives
\begin{align*}
W(u \overline {u})=m  \overline{(1+p) u}-u \overline{(1+p)m}~,
\end{align*}
which is the same as \eqref{lem1}, using $\overline{(1+p) u}= \overline{u}(1+p)$ and
$u \overline{(1+p)}= (1+p) u$.
Thus, \eqref{seq1} and \eqref{seq2} are equivalent, and all three equations --
\eqref{seq1}, \eqref{seq2} and \eqref{seq3} -- imply the equivalent
formulas \eqref{lem}, \eqref{lem1}.

Next, assume \eqref{seq2} holds.
Then, we have \cite{R-N}
\begin{align*}
q={1\over 2\pi i}
\int\limits_{\Gamma}(z-q)^{-1}\, dz \, \quad \mbox{and}\\
\sqrt{1+q} =-{1\over 2\pi i}
\int\limits_{\Gamma}(q-z)^{-1}\sqrt {1+z}\, dz \, \,,
\end{align*}
and
\begin{align}
W((q-z)^{-1})=-(q-z)^{-1} W(q)(q-z)^{-1}~, \notag\\
W(\sqrt{1+q}) ={1\over 2\pi i}
\int\limits_{\Gamma}(q-z)^{-1} W(q)(q-z)^{-1}\,\sqrt{1+z}\, dz\ . \label{cint}
\end{align}
So, \eqref{seq3} follows, since $W(p)=W(\sqrt{1+q})$ and $W(q)=F$.

Conversely, assume \eqref{seq3} holds. Then, $W(q)=F$ as before; and $W(p)$
is given by \eqref{cint}, thus \eqref{seq2} holds.

\end{proof}

\begin{theorem} \label{localex-II}
Using the same notation as in Theorem \ref{equiv},
let $u_0 \in L^2(\Bbb R^6)$ be symmetric, given. There exists $\epsilon_0$ such that if $\|m\|_{L^1([0, T]) L^2 (dx \, dy)} \le \epsilon_0$
then there exists
$u \in L^{\infty}([0, T]) L^2 (dx dy)$ solving \eqref{seq3} with prescribed initial condition $u(0, x, y)=u_0(x, y) \in L^2(\Bbb R^6)$.
The solution $u$ satisfies the following additional properties:
\begin{enumerate}
\item
\begin{align}
\| \left( i
\frac{\partial}{\partial t} -\Delta_x -\Delta_y \right) u \|_{ L^1([0, T]) L^2(dx dy)} \le C~; \label{su}
\end{align}
\item
\begin{align}
\| \left( i
\frac{\partial}{\partial t} -\Delta_x +\Delta_y \right) p \|_{ L^1([0, T]) L^2(dx dy)} \le C~. \label{wp}
\end{align}
In this context, $ p$ is defined as $\sqrt{1+ u \overline u}-1$.
\end{enumerate}
\end{theorem}
\begin{proof}
The equation \eqref{seq3} is of the form
\begin{align}
S(u)= m+ N(u)~, \label{neq}
\end{align}
where $N(u)$ involves no derivatives of $u$.
Recall the fixed time estimate $\|k l\|_{L^2} \le \|k\|_{op} \|l\|_{L^2}$,
where $op$ stands for the operator norm, and $L^2$ stands for the Hilbert-Schmidt norm. Since $r$ and $(q-z)^{-1}$ have uniformly bounded operator norms
and $\|p\|_{L^2} \le \|u\|_{L^2}$, and also $|z| \le C \|u\|^2_{L^2}$ on $\Gamma$,
we have
\begin{align*}
&\|N(u)\|_{L^1L^2} \le C
(1+ \|u\|^5_{L^{\infty}L^2})
\|m\|_{L^1L^2}~,\\
&\|N(u)-N(v)\|_{L^1L^2} \le C \max\{1,
\|u\|^4_{L^{\infty}L^2},
\|v\|^4_{L^{\infty}L^2} \}
\|m\|_{L^1L^2}~,\\
&\times
\|u-v\|_{L^{\infty}L^2}
\end{align*}
where $L^1 L^2$ stands for $ L^1([0, T])L^2(\Bbb R^6)$
and $L^{\infty} L^2= L^{\infty}([0, T])L^2(\Bbb R^6)$.
Recalling the energy estimate
\begin{align*}
\|u\|_{L^{\infty}L^2} \le \|u(0, \cdot)\|_{L^2}
+ \|Su\|_{L^1L^2}~,
\end{align*}
we see that, for any given $C$ there exists an $\epsilon_0$
 such that \eqref{neq}
 has a fixed point solution in the
set $\|u\|_{L^{\infty}L^2} \le C$ provided
$\|m\|_{L^1L^2} \le \epsilon_0$.

To prove \eqref{su}, we already know that $\|Su\|_{L^1L^2} \le C$,
so we must only account for the lower order terms in $g$, namely
$v_{12} \phi_1 \overline{\phi_2} u$ (composition of kernels)
and $(v * |\phi|^2) u$ (multiplication). These are both easy
because we know $\|u\|_{L^{\infty}L^2} \le C$ and Theorem
\ref{hest} implies $\|v_{12} \phi_1 \overline{\phi_2}\|_{L^1L^2} \le C$
as well as  $\|v * |\phi|^2\|_{L^1L^{\infty}} \le C$ , since $v \in L^2$.

A similar proof applies in order to show that \eqref{wp} follows from estimate
\eqref{lem1}.

\end{proof}

\section{Estimates for error terms}
\label{er}
In this section we obtain estimates for the error terms $\int_0^T \|e^BV e^{-B}\Omega\|_{\F}\ dt$ (quartic term)
and $\int_0^T\|e^B[A, V]e^{-B} \Omega\|_{\F}\ dt$ (cubic term). These terms were encountered in paper I.

We start by recalling the following result (Proposition 2, section 7 of paper I):

\begin{proposition} \label{vprop}
The state $e^B V e^{-B} \Omega $ has entries on the zeroth, second and fourth slot of a Fock space vector of the form given in paper I.
In addition, if
\begin{align*}
\| \left( i
\frac{\partial}{\partial t} -\Delta_x -\Delta_y \right) u \|_{ L^1[0, T] L^2(dx dy)} \le C_1,
\end{align*}
\begin{align*}
\| \left( i
\frac{\partial}{\partial t} -\Delta_x +\Delta_y \right) p \|_{ L^1[0, T] L^2(dx dy)} \le C_2
\end{align*}
and $v(x) = \chi(x)\frac{1}{|x|}$, or $v(x) = \frac{1}{|x|}$,
then
\begin{align*}
\int _0^T\|e^B V e^{-B} \Omega\|^2_{\F} \, \, dt\le C~,
\end{align*}
where $C$ only depends on $C_1$ and $C_2$.
\end{proposition}
Based on this result, estimates \eqref{su} and \eqref{wp} and Cauchy-Schwarz in time we conclude:
\begin{proposition} \label{verror}
The following estimate holds:
\begin{align}
\int_0^T
\|e^BV e^{-B}\Omega\|_{\F} dt \le C T^{1/2}~. \label{error1}
\end{align}
\end{proposition}

Now we turn attention to $\int_0^T\|e^B[A, V]e^{-B} \Omega\|_{\F} dt$, seizing
the opportunity of improving on results in section 8 of paper I.
There, we had to estimate a certain trace; see equations (61) and (62) of paper I.
This task can be avoided by commuting $a_{x_2}$ and $a_{y_2}^\ast$ in equation (60) of paper I.
Thus, terms involving  $\sh(x, x)$, as in (62) of paper I, can in fact be avoided.
To illustrate this point, we include the calculations here in the Appendix, which in effect replaces section 8 of paper I, incorporating the above remark.
Our result is now simpler and stronger.

\begin{proposition} \label{vaprop}
The state $e^B [A, V] e^{-B} \Omega $ has entries in the first and third slot of a Fock space vector of the form $\psi_{I}$ - $\psi_{III}$
 and $\psi_{I'}$ - $\psi_{III'}$ given in
  the Appendix.
In addition, if
\begin{align*}
\| \left( i
\frac{\partial}{\partial t} -\Delta_x -\Delta_y \right) u \|_{ L^1[0, T] L^2(dx dy)} \le C_1,
\end{align*}
\begin{align*}
\| \left( i
\frac{\partial}{\partial t} -\Delta_x +\Delta_y \right) p \|_{ L^1[0, T] L^2(dx dy)} \le C_2~,
\end{align*}
\begin{align}
\| \left( i
\frac{\partial}{\partial t} +\Delta_x \right) \phi \|_{ L^1[0, T] L^2(dx dy)} \le C_3~, \label{phest}
\end{align}

and $v(x) = \chi(x)\frac{1}{|x|}$, or
$v(x) = \frac{1}{|x|}$, then we have
\begin{align*}
\int _0^T\|e^B [A, V] e^{-B} \Omega\|^2_{\F} \, \, dt\le C~,
\end{align*}
where $C$ only depends on $C_1$,  $C_2$ and  $C_2$.
Thus, the following estimate holds:
\begin{align}
\int_0^T
\|e^B[A, V] e^{-B}\Omega\|_{\F} dt \le C T^{1/2}~. \label{error2}
\end{align}
\end{proposition}

\begin{remark} Notice that \eqref{phest} is satisfied by
Theorem \eqref{hest}.
\end{remark}

Estimates \eqref{error1} and \eqref{error2} form the basis of Theorem \ref{mainnewthm}, which is the main result of this paper.

\section{The operator $e^B$ \label{expb}}

In paper I, we used the definition
\begin{align}
B(t) :=\frac{1}{2} \int \left(k(t, x, y)a_x a_y - \overline k(t, x, y) a^*_x a^*_y\right)\ dx\, dy \label{beq}
\end{align}
with $\|k\|_{L^2(dx \, dy)}$ small;
$e^B$ was defined as a convergent Taylor series on the dense subset of vectors in $\F$ with finitely many nonzero components, and then it was extended to $\F$ as a unitary operator.
Consider the Lie algebra $sp (\Bbb R)$, or $sp (\Bbb C)$ of symplectic matrices with real (or complex), bounded operator coefficients.
These satisfy $J S + S^T J =0$ and have the form
\begin{align*}
S=
\left(
\begin{matrix}
a& b\\
c&- a^t
\end{matrix}
\right)
\end{align*}
where $b=b^T$, $c=c^T$. Further, consider
the corresponding groups $ Sp(\Bbb R)$,
$ Sp(\Bbb C)$ of bounded operators $G$
which satisfy $G^T J G =J$. In applications, $G=e^S \in Sp$ is defined by a convergent Taylor series.
By definition, $G$ acts on $\phi=f+ig$ by acting on the vector
$\left(
\begin{matrix}
f\\
g
\end{matrix}
\right)
$
and, of course, preserves the symplectic form $\Im \int \phi \overline \psi$.

The following
Lie algebra isomorphism from $sp(\Bbb C)$ to operators (not necessarily skew-Hermitian)
was a crucial ingredient in paper I:
\begin{align}
\left(
\begin{matrix}
d&k\\
l&-d^T
\end{matrix}
\right) &\to
\I\left(
\begin{matrix}
d&k\\
l&-d^T
\end{matrix}
\right)
:=
\frac{1}{2}\left( \begin{matrix} \label{metaplectic}
a_x& a_x^*
\end{matrix}
\right)
\left(
\begin{matrix}
d&k\\
l&-d^T
\end{matrix}
\right) J
\left(
\begin{matrix}
a_y\\
a^*_y
\end{matrix}
\right)\\
&=-
 \int d(x, y) \frac{
a_x a_y^* + a_y^* a_x}{2}\ dx\, dy \notag
+  \frac{1}{2}
\int k(x, y)a_x a_y\ dx\, dy\\
&- \frac{1}{2}
\int l(x, y)a_x^* a_y^*\ dx\, dy~. \notag
\end{align}
To ensure that the resulting operator is skew-Hermitian
we now restrict this isomorphism to the Lie subalgebra
$sp_c(\Bbb R) :=\overline C sp(\Bbb R) C^T$
for
\begin{align*}
C=\frac{1}{\sqrt 2}\left(
\begin{matrix}
1&-i\\
1&i
\end{matrix}
\right)~.
\end{align*}
This is a change of basis that will be explained below.

\begin{lemma} The map
\begin{align}
S \mapsto  \overline C S C^T \label{gpiso}
\end{align}
is a Lie algebra isomorphism of $sp(\Bbb C)$ to $sp(\Bbb C)$.
\end{lemma}

\begin{proof}
The "matrix" $C$ is unitary ($C^T = ( \overline C)^{-1}$ )
and also satisfies $C^T J C=C J C^T= i J$; thus,
$i^{-\frac{1}{2}} C$ and
 $i^{-\frac{1}{2}} C^T$
belong to the symplectic group $Sp(\Bbb C)$ (and the choice of
$i^{-\frac{1}{2}}$ does not matter). Since
\eqref{gpiso} does not change if we replace $C$ by
$i^{-\frac{1}{2}} C$, we see that \eqref{gpiso}  is just conjugation
by an element of $Sp(\Bbb C)$, and thus is a Lie algebra isomorphism.
\end{proof}

\begin{lemma}
If $S \in sp_c(\Bbb R)$, then $\I(S)$ is skew-Hermitian.
\end{lemma}

\begin{proof}
This proof will also motivate the choice of $C$.

 Define the self-adjoint operators of ``momentum''
\begin{align*} P_x := D_x =\frac{a_x+a^*_x}{\sqrt 2}
\end{align*}
 and ``position''
\begin{align*}Q_x := X_x =\frac{i(a_x- a^*_x)}{\sqrt 2}~.
 \end{align*}
These satisfy the canonical relations
\begin{align*}
[D_x, X_y] = \frac{1}{i}\delta(x-y)~.
\end{align*}

We will rewrite \eqref{metaplectic} in terms of the self-adjoint operators $D$ and $X$.
The change-of-basis formula is
\begin{align}
\left(
\begin{matrix}
a_x\\
a^*_x
\end{matrix}
\right)
= C
\left(
\begin{matrix}
D_x\\
X_x
\end{matrix}
\right)
\label{C-change}
\end{align}
for
\begin{align*}
C=\frac{1}{\sqrt 2}\left(
\begin{matrix}
1&-i\\
1&i
\end{matrix}
\right)~;
\end{align*}
see page 174, (4.13) of \cite{F} for a closely related construction.
Notice that $ JC= i \overline C J$ with $\overline C = (C^T)^{-1}$;
thus,
\begin{align*}
&\frac{1}{2}\left( \begin{matrix}
a_x& a_x^*
\end{matrix}
\right)
\left(
\begin{matrix}
d&k\\
l&-d^T
\end{matrix}
\right) J
\left(
\begin{matrix}
a_y\\
a^*_y
\end{matrix}
\right)\\
&=
\frac{i}{2}\left( \begin{matrix}
D_x& X_x
\end{matrix}
\right) C^T
\left(
\begin{matrix}
d&k\\
l&-d^T
\end{matrix}
\right) \overline C J
\left(
\begin{matrix}
D_y\\
X_y
\end{matrix}
\right)~.\\
\end{align*}
At this point it is natural to introduce the Lie algebra isomorphism
$sp(\Bbb C) \to sp(\Bbb C)$,
\begin{align}
A= A_{a, \, a^*} \to A_{D, \, X} := C^T A_{a, \, a^*} \overline C~. \label{iso}
\end{align}
Since $C$ is unitary ($C^T=(\overline C)^{-1}$),
this is the inverse of \eqref{gpiso}.
At this stage it is clear that if
\begin{align*}
C^T
\left(
\begin{matrix}
d&k\\
l&-d^T
\end{matrix}
\right) \overline C
\end{align*}
is real then the corresponding operator is skew-Hermitian.
   Thus, $sp_c(\Bbb R)$ consists of
   those $A_{a, \, a^*}$
 such that the corresponding  $A_{D, \, X} \in sp(\Bbb R)$.
\end{proof}

\begin{remark}

In particular, for our $K$,
\begin{align*}
K=
\left(
\begin{matrix}
0&k\\
\overline k&0
\end{matrix}
\right)~,
\end{align*}
the corresponding decomposition in $D_x$ and $X_x$ is (see \eqref{C-change})
\begin{align}
K_{D, \, X}=C^TKC=\left(
\begin{matrix}
\Re k& \Im k\\ \label{chbasis}
\Im k&- \Re k
\end{matrix}
\right)~;
\end{align}
thus, $K \in sp_C(\Bbb R)$.
\end{remark}

It is easy to check that, if $S \in sp_c(\Bbb R)$, then
\begin{align*}
 e^S\left(
\begin{matrix}
\overline{\phi}\\
- \phi
\end{matrix}
\right) \, \mbox{is of the form} \,
\left(
\begin{matrix}
\overline{\psi}\\
- \psi
\end{matrix}
\right)~.
\end{align*}
Thus, it is legitimate to parametrize the vector
$\left(
\begin{matrix}
\overline{\phi}\\
- \phi
\end{matrix}
\right)$
by $\phi$ and denote
\begin{align*}
e^S(\phi):=e^S\left(
\begin{matrix}
\overline{\phi}\\
- \phi
\end{matrix}
\right)~.
\end{align*}
We also define $A
\left(
\begin{matrix}
f\\
g
\end{matrix}
\right) =a(f) + a^*(g)$
so that
$A(\phi):=
A
\left(
\begin{matrix}
\overline{\phi}\\
- \phi
\end{matrix}
\right)$.

We now recall the results of section 4 in paper I:
\begin{theorem} Let $\phi \in L^2$ and $R, S \in sp(\Bbb C)$ with $L^2$
(or Hilbert-Schmidt) coefficients. Then

\begin{equation}
[\I(S), A\left(
\begin{matrix}
f\\
g
\end{matrix}
\right) ]=A(S \left(
\begin{matrix}
f\\
g
\end{matrix}
\right))~, \label{lie1}
\end{equation}
and therefore
\begin{align}
&[\I(S), A(\phi)]=A(S(\phi))~, \\
&[\I(S), \I(R)] =\I[S, R]~. \label{lie2}
\end{align}

In addition, if $S \in sp_C(\Bbb R)$ and $\|S\|_{L^2}$ is small, then
\begin{align}
&e^{\I(S)} A
\left(
\begin{matrix}
f\\
g
\end{matrix}
\right) e^{-\I(S)}=A(e^S\left(
\begin{matrix}
f\\
g
\end{matrix}
\right))~, \label{grp1}\\
&e^{\I(S)} A(\phi) e^{-\I(S)}=A(e^S(\phi))~, \\
 &e^{\I(S)} I(R) e^{-\I(S)}=\I(e^S R e^{-S})~, \label{grp2}\\
& \left(\frac{\partial}{\partial t} e^{\I(S)} \right) e^{-\I(S)}=
 \I\left(\left(\frac{\partial}{\partial t} e^S \right) e^{-S}\right)\ . \label{grp3}
 \end{align}

  \end{theorem}
\begin{proof} The formulas \eqref{lie1}--\eqref{lie2} are elementary calculations. Formulas \eqref{grp1}--\eqref{grp3} follow by
analyticity (power series) since $e^{\I(S)}$ is given by a convergent Taylor series on the dense subset of Fock space vectors with finitely many
non-zero components. Replace $S $ by $tS$ ($t \in \Bbb C$, small) and check that all derivatives of the left-hand side agree with all derivatives
of the right-hand side at $t=0$.
\end{proof}

For $\|S\|_{L^2(dx \, dy)}$ large, $S \in sp_c(\Bbb R)$, the series
defining $e^{\I(S)}$ may not converge on a dense subset.
So, we define
\begin{align*}
e^{\I(S)} =\left(e^{\I(S)/n}\right)^n~,
\end{align*}
where $n$ is so large that $e^{\I(S)/n}$ is defined by a convergent series on vectors with finitely many components, and is then extended as a unitary operator to $\F$.
This definition is clearly independent of $n$ and still satisfies the crucial
properties \eqref{grp1}--\eqref{grp3}.

For the rest of this section, we discuss connections with well-known results and explain the change-of-basis formula.

\subsection{Connection to the Heisenberg group and metaplectic representation}

Recall that the classical Heisenberg group ${\bf H}_n$ is $\Bbb C^n \times \Bbb R$ with multiplication law $(z, t)  (w, s) = (z+w, t+s - \Im z \overline w)$;
see (1.20) in \cite{F}.
In our setting, ${\bf H}$ is
$L^2 (\Bbb R^3 ) \times \Bbb R$
with multiplication law $(\phi, t)  (\psi, s) = (\phi + \psi, t+s - \Im \phi \overline \psi)$.
The map
$ (\phi, t) \to e^{-A(\phi)} e^{it}$ is a unitary representation of ${\bf H}$.
Indeed, we have
\begin{align*}
&e^{-A(\phi)} e^{it} e^{-A(\psi)} e^{is}
=e^{-A(\phi+ \psi) + \frac{1}{2}[A(\phi), A(\psi)]} e^{it+ is}\\
&=e^{-A(\phi+ \psi) + \frac{1}{2}\int \left(\overline \phi \psi - \phi \overline \psi \right)} e^{it +is}\\
&=e^{-A(\phi+ \psi)}   e^{i(t +s -\Im \int\ \phi \overline \psi)}~.
\end{align*}

Shale \cite{shale} extended the standard construction of the metaplectic representation (see chapter 4 in \cite{F}) to the infinite dimensional
``restricted symplectic group''
$rSp (\Bbb R) =\{T \in Sp(\Bbb R), (T^* T)^{1/2} -I \, \mbox{is Hilbert-Schmidt} \} $.
We do not use his results directly; and the following comments
are just meant for completeness.
His results, and those of \cite{F},
are written with respect to the basis $D_x, X_x$.
By assuming $G \in rSp (\Bbb R)$, Shale showed there exists a unitary transformation of $\F$, $Y(G)$,  such that
\begin{align}
e^{-A(G \phi)} = Y(G) e^{-A(\phi)} Y(G)^{-1}~; \label{conj}
\end{align}
also, any two such unitary transformations $Y_1(G)$, $Y_2(G)$
are related by $Y_1(G)= e^{i \theta} Y_2(G)$. The mapping
$G \mapsto Y(G)$ is a projective unitary representation, meaning that
$Y(G_1) Y(G_2) = e^{i \theta(G_1, G_2)} Y(G_1 G_2)$.
In particular, we identify our unitary operator
$e^{\I(S)}$
(after we reconcile the bases) as
$e^{\I(S)}= e^{i \theta} Y( e^S)$ for some
 $\theta= \theta(S ) \in \Bbb R$; we skip further details.

\appendix

\section{Computation of cubic error term}

With recourse to equation (56) of paper I, and because of the comments following Proposition \eqref{verror}, we now carefully compute
the error term
\begin{align}
e^B[A, V]e^{-B}=
\int v(x-y) \Big( &\overline \phi (y)e^B a^*_xe^{-B} e^B a_xe^{-B}e^B a_y
e^{-B} \label{E1}\\
& + \phi (y) e^B a^*_x e^{-B}e^Ba^*_y e^{-B} e^B a_x e^{-B} \Big)\ dx\, dy~,
\label{E2}
\end{align}
which acts on the vacuum state, $\Omega$. All terms ending in $a$ can be ignored. After commuting all $a$ terms to the right,
we are left with a pure cubic
and a pure linear term
in $a^*$, which we proceed to compute.
Recall the following formula proved in paper I:
\begin{align}
e^B (a_y, a^*_y)
\left(
\begin{matrix}
f\\
g
\end{matrix}
\right)e^{-B}=(a_y, a^*_y)e^K\left(
\begin{matrix}
f\\
g
\end{matrix}
\right)~.
\end{align}
Thus, we have
\begin{align*}
e^B a_x e^{-B}
=
\int \left( \ch (y, x) a_y + \shb (y, x) a^*_y \right)\ dy~,
\end{align*}
and, similarly,
\begin{align*}
e^B a^*_x e^{-B}=
\int \left( \sh (y, x)a_y + \chb (y, x) a^*_y \right) dy~.
\end{align*}

We are ready to extract the pure $a^*$ term from \eqref{E1}.
Before any simplifications, \eqref{E1} reads
\begin{align*}
\int v(x-y)  \overline \phi (y)&
 \left( \sh (z_1, x)a_{z_1} + \chb (z_1, x)a^*_{z_1} \right)\\
 &
\left( \ch (z_2, x)a_{z_2} +\shb (z_2, x) a^*_{z_2} \right)\\
&\left( \ch (z_3, y)a_{z_3} + \shb (z_3, y)a^*_{z_3} \right)\,
dz_1 dz_2 dz_3 dx dy~.
\end{align*}
Thus, \eqref{E1} contributes the cubic term
\begin{align*}
I=\int v(x-y)  \overline \phi (y)
 \left(\chb (z_1, x)
\shb (z_2, x)
\shb (z_3, y) \right)a^*_{z_1} a^*_{z_2}a^*_{z_3}\\
dz_1 dz_2 dz_3 dx dy~. \notag
\end{align*}
Thus $
I(\Omega)$ has entries in the third slot of Fock space equal to
(after normalization and symmetrization)
\begin{align*}
\psi_I(z_1, z_2, z_3)=\int v(x-y)  \overline \phi (y)
 \left(\chb (z_1, x)
\shb (z_2, x)
\shb (z_3, y) \right) dx dy~. \notag
\end{align*}
For the linear terms, keep only the $a^*$ term from the last row, and exactly one $a$ and one $a^*$s from the first and second rows, and commute the $a$'s to the right.
Hence, we are left with two terms:
\begin{align*}
II=&\int v(x-y)  \overline \phi (y)
 \sh (z_1, x)
 \shb (z_2, x)  \shb (z_3, y) \notag
 a_{z_1} a^*_{z_2}
 a^*_{z_3}
dz_1 dz_2 dz_3 dx dy\\
&=
\int v(x-y)  \overline \phi (y)
 \sh (z_1, x)
 \shb (z_2, x)  \shb (z_3, y) \notag\\
& \left(\delta(z_1 - z_3)
 a^*_{z_2} +\delta(z_1 - z_2)
 a^*_{z_3}\right)
 dz_1 dz_2 dz_3 dx dy \, \,
\mbox{(modulo linear terms in $a$)} \notag \\
&=
\int v(x-y)  \overline \phi (y)
 \sh (z_1, x)
 \shb (z_2, x)  \shb (z_1, y) \notag
 a^*_{z_2}
 dz_1 dz_2 dx dy \\
&+
\int v(x-y)  \overline \phi (y)
 \sh (z_1, x)
 \shb (z_1, x)  \shb (z_3, y)
 a^*_{z_3}
 dz_1 dz_3 dx dy~,
 \end{align*}
and
\begin{align*}
III&=\int v(x-y)  \overline \phi (y) \notag
 \chb (z_1, x)
  \ch (z_2, x) \shb (z_3, y)
  a^*_{z_1} a_{z_2}
  a^*_{z_3}
dz_1 dz_2 dz_3 dx dy\\
&=\int v(x-y)  \overline \phi (y)
 \chb (z_1, x)
  \ch (z_2, x) \shb (z_2, y)\notag
  a^*_{z_1}
dz_1 dz_2  dx dy~.\\
\end{align*}
These terms contribute to the first slot of Fock space entries of the form
\begin{align*}
\psi_{II}(z) =&\int v(x-y)  \overline \phi (y)
 \sh (z_1, x)
 \shb (z, x)  \shb (z_1, y) \notag
 dz_1  dx dy \\
&+
\int v(x-y)  \overline \phi (y)
 \sh (z_1, x)
 \shb (z_1, x)  \shb (z, y)
 dz_1  dx dy
 \end{align*}
and
\begin{align*}
\psi_{III}(z)=
&=\int v(x-y)  \overline \phi (y)
 \chb (z, x)
  \ch (z_2, x) \shb (z_2, y)\notag
 dz_2  dx dy~.\\
\end{align*}

Next, we concentrate on the contributions of \eqref{E2}.
\begin{align*}
\eqref{E2}=\int v(x-y)  \phi (y)&
 \left( \sh (z_1, x)a_{z_1} + \chb (z_1, x)a^*_{z_1} \right)\\
 &\left( \sh (z_2, y)a_{z_2} + \chb (z_2, y)a^*_{z_2} \right)\\
&\left( \ch (z_3, x)a_{z_3} + \shb (z_3, x)a^*_{z_3} \right)
dz_1 dz_2 dz_3 dx dy~,
\end{align*}
which contributes a cubic in $a^*$:
\begin{align*}
I'=\int v(x-y)  \phi (y)
\chb (z_1, x)
 \chb (z_2, y)
 \shb (z_3, x) a^*_{z_1} a^*_{z_2} a^*_{z_3}
dz_1 dz_2 dz_3 dx dy~.
\end{align*}
Here, we have
\begin{align*}
\psi_{I'}(z_1, z_2, z_3)=\int v(x-y)  \phi (y)
\chb (z_1, x)
 \chb (z_2, y)
 \shb (z_3, x)  dx dy~.
\end{align*}

The linear terms in $a^*$ (modulo $a$) are
\begin{align*}
&II'=\int v(x-y)  \phi (y)
 \sh (z_1, x)a_{z_1}
  \chb (z_2, y)a^*_{z_2}
\shb (z_3, x)a^*_{z_3}
dz_1 dz_2 dz_3 dx dy\\
&=
\int v(x-y)  \overline \phi (y)
 \sh (z_1, x)
 \chb (z_2, y)  \shb (z_1, x) \notag
 a^*_{z_2}
 dz_1 dz_2 dx dy \\
&+
\int v(x-y)  \overline \phi (y)
 \sh (z_1, x)
 \chb (z_1, y)  \shb (z_3, x)
 a^*_{z_3}
 dz_1 dz_3 dx dy
 \end{align*}
 and
\begin{align*}
III'&=\int v(x-y)  \overline \phi (y) \notag
 \chb (z_1, x)
  \sh (z_2, y) \shb (z_3, x)
  a^*_{z_1} a_{z_2}
  a^*_{z_3}
dz_1 dz_2 dz_3 dx dy\\
&=\int v(x-y)  \overline \phi (y)
 \chb (z_1, x)
  \sh (z_2, y) \shb (z_2, x)\notag
  a^*_{z_1}
dz_1 dz_2  dx dy~.\\
\end{align*}
The Fock space entries read
\begin{align*}
\psi_{II'}(z)=
&
\int v(x-y)  \overline \phi (y)
 \sh (z_1, x)
 \chb (z, y)  \shb (z_1, x) \notag
 \,dz_1 dx dy \\
&+
\int v(x-y)  \overline \phi (y)
 \sh (z_1, x)
 \chb (z_1, y)  \shb (z, x)
 \,dz_1  dx dy
 \end{align*}
 and
\begin{align*}
\psi_{III'}(z)
&=\int v(x-y)  \overline \phi (y)
 \chb (z, x)
  \sh (z_2, y) \shb (z_2, x)\notag
\,dz_2  dx dy~.\\
\end{align*}

All the resulting $\psi$ can be estimated in $L^2 (dt dz)$ by the method of section 7 of paper I, without using $X_{s, \delta}$ spaces.
We remind the reader how to estimate these terms. Take, for instance, $\psi_{III'}(t, z)$.
Write $\ch (t, z, x) = \delta (z-x) + p(t, z, x)$ to express
$\psi_{III'}= \psi_{\delta} + \psi_p$. We estimate the first of these terms:
\begin{align*}
|\psi_{\delta}(z)|
&=|\int v(z-y)  \overline \phi (y)
  \sh (z_2, y) \shb (z_2, z)\notag
dz_2   dy|\\
&\le \|v(z-y)  \overline \phi (y)\shb (z_2, z)\|_{L^2(d z_2 dy)}
\|\sh (z_2, y)\|_{L^2(d z_2 dy)}~.
\end{align*}
The second term is uniformly bounded in time; thus,
\begin{align*}
&\int_0^{\infty} \int_{\Bbb R^3} |\psi_{\delta}(t, z)|^2 dt dz\\
&\le
C \int_0^{\infty} \int|
v(z-y)   \phi (y)\shb (z_2, z)|^2 dz dz_2 dy dt
\le C
\end{align*}
by a local smoothing type result (see Lemma 2, section 7 of paper I).
All other terms can he estimated by the same method.

\end{document}